\documentclass[11pt]{article}

\usepackage[dvips]{graphics}  
\usepackage[T1]{fontenc}
\usepackage{array}
\usepackage{floatrow}
\usepackage{amsfonts,dsfont}
\usepackage{amsmath}
\usepackage{amssymb}
\RequirePackage[colorlinks,citecolor=blue,urlcolor=blue,linkcolor=blue]{hyperref}
\usepackage{graphicx,xspace,colortbl}
\usepackage{amsmath,amsthm,amsfonts,latexsym,amssymb}
\usepackage{color}
\usepackage{fancybox}
\usepackage{subfig}
\usepackage{pdfsync}
\usepackage{caption}
\usepackage{multirow}
\hypersetup{
colorlinks = true,
citecolor=blue,
urlcolor=blue,
linkcolor=blue,
pdfauthor = {Zied Ben-Salah, Helene Guerin, Manuel Morales, Hassan Omidi},
pdfkeywords = {Drawdowns, depletion, Levy insurance risk process, spectrally negative, scale function, risk theory},
pdftitle = {On the Depletion Problem for an Insurance Risk Process: New Non-ruin Quantities in Collective Risk Theory},
pdfpagemode = UseNone
}

\newtheorem{thm}{Theorem}
\newtheorem{lem}[thm]{Lemma}
\newtheorem{pro}{Proposition}

\newtheorem{defn}[thm]{Definition}

\newtheorem{rem}[thm]{Remark}

    \def\beq{\begin{eqnarray}}
    \def\eeq{\end{eqnarray}}
    \def\beqq{\begin{eqnarray*}}
    \def\eeqq{\end{eqnarray*}}

\newcommand{\PAR}[1]{{\left(#1\right)}} 

\title {Optimal Portfolio Problem Using Entropic Value at Risk: When the Underlying Distribution is Non-Elliptical}

\author{Hassan Omidi Firouzi \footnote{Hassan Omidi Firouzi. Department of Mathematics and Statistics. University of Montreal. CP. 6128 succ. centre-ville. Montreal, Quebec. H3C 3J7. CANADA. Email: omidifh@dms.umontreal.ca}\\
University of Montreal \and Andrew Luong \footnote{Andrew Luong. School of Actuarial Science. Laval University. Pavillon Paul-Comtois, local 4105. Quebec 
City, Quebec. CANADA. Email: andrew.luong@act.ulaval.ca}
\\ Laval University\\}
\date{{\scriptsize First draft: May~24, 2013. This version: \today.}}
\begin{document}
\maketitle

\begin{abstract}
This paper is devoted to study the optimal portfolio problem. Harry Markowitz's Ph.D. thesis prepared the ground for the mathematical theory of finance \cite{Markowitz2}. In modern portfolio theory, we typically find asset returns that are modeled by a random variable with an elliptical distribution and the notion of portfolio risk is described by an appropriate risk measure. 
In this paper,  we propose new stochastic models for the asset returns that are based on {\it Jumps- Diffusion} {\bf (J-D)} distributions \cite{Press, Ruijter}. This family of distributions are more compatible with stylized features of asset returns. On the other hand, in the past decades, we find attempts in the literature to use well-known risk measures, such as {\it Value at Risk} and {\it Expected Shortfall}, in this context. Unfortunately, one drawback with these previous approaches is that no explicit formulas are available and numerical approximations are used to solve the optimization problem. In this paper, we propose to use a new coherent risk measure, so-called, {\it Entropic Value at Risk(EVaR)} \cite{Ahmadi}, in the optimization problem. For certain models, including a jump-diffusion distribution, this risk measure yields an explicit formula for the objective function so that the optimization problem can be solved without resorting to numerical approximations. \\

 \end{abstract}
 
\section{Introduction}
The problem of optimal portfolio which is nowadays introduced in a new framework, called {\it Modern Portfolio Theory (MPT)},  has been extensively studies in the past decades. The MPT is one of the most important problems in financial mathematics. Harry Markowitz \cite{Markowitz2} introduced a new approach to the problem of optimal portfolio so called {\it Mean-Variance} analysis. He chose a preferred portfolio by taking into account the following two criteria. The expected portfolio return and the variance of the portfolio return. In fact, Markowitz preferred one portfolio to another one if it has higher expected return and lower variance.  \\

Later, we find attempts in the literature to replace variance with well-known risk measures, such as {\it Value at Risk} and {\it Expected Shortfall}.  For instance, Embrechts et al.\cite{Embrechts} have shown that replacing mean-variance with any other risk measure having the translation invariant and positively homogeneous properties under elliptical distributions yields to the same optimal solution.
Basak and Shapiro \cite{Shapiro} studied an alternative version of Markowitz problem by applying VaR for controlling the incurred risk in an   expected utility maximization framework which allows to maximize the profit of the risk takers. Studying the Markowitz model has been done in the same framework by considering the CVaR as risk measure. \cite{Rachev}. Later, Acerbi and Simonetti \cite{Acerbi} studied the same problem as the one studied in \cite{Shapiro} with spectral risk measures. Recently,  Cahuich and Hernandez \cite{Hernandez} solved the same problem within the framework
of utility maximization using the class of distortion risk measures \cite{Tsukahara}.
\\

 There are both practical and theoretical weaknesses that can be made about the relevant framework of optimal portfolio problem in the literature. One of such criticisms relates to the asset returns model itself.  In fact, elliptical distribution is the most and relevant distribution which is used to model asset returns in MPT.  One of the reason for choosing this distribution ties with the tractability of this class of distribution. But,  in practice financial returns do not follow an elliptical distribution. A second objection focuses in the choice of a measure of risk for the portfolio. Unfortunately, one drawback with the previous works, for instance \cite{Shapiro, Rachev}, is that no explicit formulas are available and numerical approximations are used to solve the optimization problem. The stochastic models which we are proposing for the asset returns in this paper are based on {\it Jumps- Diffusion} {\bf (J-D)} distributions \cite{Press, Ruijter}. This family of distributions are more compatible with stylized features of asset returns and also allows for a straight-forward statistical inference from readily available data. We also tackle the second issue by choosing a suitable (coherent) risk measure as our objective function. In this paper, we propose to use a new coherent risk measure, so-called, {\it Entropic Value at Risk(EVaR)} \cite{Follmer, Ahmadi}, in the optimization problem.  As this risk measure is based on Laplace transform of asset returns, applying it to the jump-diffusion models yields an explicit formula for the objective function so that the optimization problem can be solved without using numerical approximations.   \\

The organization of this paper is as follows. In Section 2, we 
 provide a summary of properties about coherent risk measures and {\it Entropic Value at Risk} measure. We also continue this section by presenting a typical representation of optimal portfolio problem where we minimize the risk of the portfolio for a given level of portfolio return. In Section 3, we introduce our two models to fit as asset returns and we apply them into the optimization problem. We also derive some distributional properties for these models and finish Section 3 by discussing about the KKT conditions and  optimal solutions. In Section 4,  we discuss about parameters estimation method which we have used in this paper. We also provide a numerical example for three different stocks and analyze the efficient frontiers for EVaR, mean-variance and VaR for these three stocks. In this paper we use optimization package in MATLAB to do the computations.

\section{Preliminaries }
\subsection{Coherent Risk Measures}
We are considering $L^\infty$ as the set of all bounded random variables representing financial positions. The following definition is taken from \cite{Follmer2}.
\begin{defn}. A function $\rho:L^\infty\rightarrow \mathbb{R}$ is
 a Coherent Risk measure if 
\item[1-]$\rho(\lambda X+(1-\lambda)Y)\leq \lambda\rho(X)+(1-\lambda)\rho(Y)$ for any $X,Y\in
 L^\infty$ and $\lambda\in[0,1]$.(Convexity)
\item[2-]$\rho(\lambda X)=\lambda \rho(X)$ for any $X\in L^\infty$ and
 $\lambda>0$.(Positive Homogeneity)
\item[3-]$\rho(X+m)=\rho(X)-m$ for any $X\in L^\infty$ and
 $m\in\mathbb{R}$.(Translation Invariant)
\item[4-]$\rho(Y)\leq\rho(X)$ $\forall X,Y \in L^\infty$ and $X\leq
 Y$.(Decreasing) 
\end{defn}
In this paper, we propose to use the {\it Entropic Value at Risk} measure (EVaR${}_{\alpha}$) which is a coherent risk measure. Following \cite{Ahmadi} we now give a first definition.
  
  \begin{defn}\label{def:EVAR}
Let $X$ be a random variable in $L^{\infty}(\Omega, \mathcal F)$ such that 
$$\mathbb E [\exp(-s\,X)]<\infty \;,\qquad s >0\;.$$
Then the {\bf Entropic Value at Risk}, denoted by EVaR${}_{\alpha}$,  is given by 
\begin{equation}\label{EVAR}
EVaR_{\alpha}(X) := \mathop{\inf}_{s>0} \frac{\ln \mathbb{E}[\exp(-s\,X)] - \ln\alpha}{s},
\end{equation}
For a given level $\alpha\in(0,1)$.
\end{defn}

The following key result for EVaR${}_{\alpha}$ can be found in \cite{Follmer, Ahmadi}.
\begin{thm}
The risk measure EVaR${}_{\alpha}$ from Definition \ref{def:EVAR} is a {\it coherent} risk measure. Moreover, for any $X\in L^{\infty}(\Omega, \mathcal F)$ having Laplace transform, its dual representation has the form
\begin{equation*}\label{Robust}
EVaR_{\alpha}(X) = \mathop{\sup}_{f\in \mathcal{D}} \mathbb E_{\mathbb P}(-fX)\;,
\end{equation*}
where $\mathcal{D} = \{ f\in L^1_+(\Omega,\mathcal F) \;|\; \mathbb E_{\mathbb P}[f \ln(f)] \leq - \ln\alpha\}$
 and 
\begin{equation*}\label{set}
L^1_+(\Omega,\mathcal F) = \{f \in L^1(\Omega,\mathcal F)  \;|\; \mathbb E_{\mathbb P}(f) =1\}.
\end{equation*} 
 \end{thm}
For a comprehensive study on this risk measure we may refer to \cite{Ahmadi, Follmer}.
\subsection{Optimal Portfolio Problem}
Consider a portfolio in a financial market with $n$ different assets. Denote the assets returns by the vector $ R= ( R_1,\dots, R_n)$ in which $R_i$ shows the return of the i-th asset. The returns are random variables and their mean is denoted by  $\mu =(\mu_1,\dots, \mu_n)$ where $\mu_i$ is the the expected return of the i-th asset, $\mu_i = \mathbb{E}(R_i)$. Moreover, assume $\rho$ as a risk measure. Then following \cite{Rachev}
\begin{defn}\label{general model}
the {\it optimal portfolio problem} can be written mathematically as follows.
 \begin{eqnarray}\label{optimal1}
\min_{\omega} & & \rho(\sum_{i=1}^n \omega_i R_i)\nonumber\\
\text{subject to}& &\sum_{i=1}^n \omega_i \mu_i =\mu^*,\nonumber\\
& &\sum_{i=1}^n \omega_i=1,\nonumber\\
& &\omega_i \geq 0,
\end{eqnarray}
\noindent where $\mu^*$ is a given level of return.
\end{defn}
Applying various risk measures along with different models for random returns yields to interesting problems in both theoretical and practical point of views. For instance, the classical mean-variance model introduced by Markowitz \cite{Markowitz2} is a special case of the model introduced in Definition $\ref{general model}$. In fact, Markowitz used variance as a risk measure and apply it into the objective function given in \eqref{optimal1} and he also considered returns from the portfolio are normally distributed.
\begin{rem}\label{rem}
It has been shown in \cite{Embrechts} that if we assume the return variables follow elliptical distributions(like multivariate normal distribution), then the solution for the Markowitz mean-variance problem will be the same as the optimal solution for optimal portfolio problem \eqref{optimal1} by minimizing any other risk measure having the translation invariant and positively homogeneous properties for a given level of return.
\end{rem}

\cite{Hu} has shown in his PhD thesis that for two different examples of elliptical distributions(normal and Student t) the portfolio decomposition for Expected Shortfall and Value at Risk are the same as the one for standard deviation.\\

Referring to Remark \ref{rem} we see that if the underlying distribution is elliptical, then for any coherent risk measure the optimal solution for the problem in \eqref{optimal1} is the same as the optimal solution for the classical model by Markowitz. 

\section{Set up the Models }
In this section, we propose two multivariate models which do not follow elliptical distributions. These models which are based on jump-diffusion distributions can be fitted as the underlying models for returns. Distributional properties of these models will be also studied.
\subsection{Non-Elliptical Multivariate Models 1,2}
{\bf Multivariate Model 1.}
Consider the following multivariate model:

\begin{equation}\label{optimal2}
R=  X + H + \sum_{k=1}^M W_k,
\end{equation}

\noindent where $R, X, H, W_k$ are $n$-variate vectors such that 

\begin{eqnarray*}
R &=&  \left( R_1,\dots, R_n\right),\\
X&=& \left(X_1,\dots, X_n \right),\\
W_k&=&\left( W_{k1},\dots, W_{kn}\right),\\
H&=& \left ( \sum_{k=1}^{N_1} Y_{k1},\dots,\sum_{k=1}^{N_n} Y_{kn} \right). \\
\end{eqnarray*}

Here, $X_i$ follows the normal distribution with $X_i \sim N(\tilde{\mu_i}, \sigma^2)$ and $X_i$'s are mutual independent for $i = 1,\dots, n$. $W_k=\left( W_{k1},\dots, W_{kn}\right)$ is assumed to follow the multivariate normal distribution with $W_k \sim N(\mu, A)$ for each $k$ where $\mu = (\mu_1,\dots, \mu_n)$ is mean and $A$ is covariance matrix. Moreover, $W_k$'s are assumed to be mutually independent. The random variable $M$ follows the Poisson distribution with intensity $\gamma$ and is independent of $W_k$ for each $k$.  $N_k$ are assumed to have Poisson distribution with intensity $\lambda_k$ and mutually independent for $k=1,\dots, n$. The $Y_{ki}$ are assumed to be mutually independent for all k and all $i= 1,\dots, n$ and $Y_{ki}$ is normal distributed with $Y_{ki} \sim N(\theta_i, \sigma_i^2)$. Finally, $N_k$ and $Y_{kn}$ are mutually independent  as well as $X , H, \sum_{k=1}^M W_k$.\\

This model can be driven from a jump-diffusion model which is the solution for a stochastic differential equation \cite{Press}.
 We can rewrite this multivariate model as follows.

\begin{eqnarray*}
R_1&=&  X_1 + \sum_{k=1}^{N_1} Y_{k1} + \sum_{k=1}^MW_{k1},\nonumber\\
R_2&=&  X_2 + \sum_{k=1}^{N_2} Y_{k2} + \sum_{k=1}^MW_{k2},\nonumber\\
\vdots\nonumber\\
R_n&=&  X_n + \sum_{k=1}^{N_n} Y_{kn} + \sum_{k=1}^MW_{kn}.\\
\end{eqnarray*}

{\bf Multivariate Model 2.}
The model  $(6.5)$ in \cite{Ruijter} prepared the ground to introduce another non-elliptical multivariate model which can be fitted for portfolio returns. This proposed model is given as follows.  
\begin{equation}\label{optimal5}
R= X + \sum_{k=1}^M W_k.
\end{equation}
\noindent Here, $R, X, W_k$ are $n$-variate vectors such that 

\begin{eqnarray*}
R &=&  \left( R_1,\dots, R_n\right),\\
X&=& \left(X_1,\dots, X_n \right),\\
W_k&=&\left( W_{k1},\dots, W_{kn}\right),\\
\end{eqnarray*}
\noindent where $X = (X_1,\dots, X_n)$ follows the multivariate normal distribution with $X\sim N(\tilde{\mu}, Q)$ with covariance matrix $Q$. $W_k=( W_{k1},\dots, W_{kn})$ is assumed to follow the multivariate normal distribution with $W_k \sim N(\mu, A)$ for each $k$ where $\mu = (\mu_1,\dots, \mu_n)$ is mean and $A$ is covariance matrix. Moreover, $W_k$'s are assumed to be mutually independent. The random variable $M$ follows the Poisson distribution with intensity $\lambda$ and is independent of $W_k$ for each $k$. Also, $X , \sum_{k=1}^M W_k$ are mutually independent.\\

Like the model \eqref{optimal2} introduced in Subsection 3.1 we can rewrite the multivariate model \eqref{optimal5} as 

\begin{eqnarray*}
R_1&=&  X_1 + \sum_{k=1}^MW_{k1},\nonumber\\
R_2&=&  X_2 + \sum_{k=1}^MW_{k2},\nonumber\\
\vdots\nonumber\\
R_n&=&  X_n + \sum_{k=1}^MW_{kn}.\\
\end{eqnarray*}

\subsection{Distributional Properties of the Multivariate Models 1, 2}
Consider the multivariate models \eqref{optimal2} and \eqref{optimal5}. As these models are given in terms of summation of multivariate normal and compound Poisson distributions we can provide the joint density functions for each of these models. \cite{Press2} gives the following presentation for the density function of model \eqref{optimal2} and also provides a proof but we give a proof here for the sake of completeness.
\begin{pro}\label{density1}
Consider the model \eqref{optimal2}. Then the joint density functions of the vector $R$ is given by
\begin{equation}\label{density}
f_R(r)= \sum_{k_1=0}^{\infty}\dots  \sum_{k_n=0}^{\infty}  \sum_{m=0}^{\infty}
\left( \frac{e^{-\lambda_1} \lambda_1^{k_1}}{k_1!}\right) \dots \left( \frac{e^{-\lambda_n} \lambda_n^{k_n}}{k_n!}\right) \left( \frac{e^{-\gamma} \gamma^{m}}{m!}\right) \frac{e^{-\frac12(r-u)T^{-1}(r-u)'}}{(2\pi)^{\frac n2}|T|^{\frac 12}},
\end{equation}
where $r=(r_1,\dots,r_n) \in \mathbb{R}^n$, $u=(\tilde{\mu_1}+k_1\theta_1,\dots, \tilde{\mu_n}+k_n\theta_n) + m\mu$ and $T=mA+diag(\sigma^2+k_1 \sigma_1^2,\dots,\sigma^2+k_n \sigma_n^2  )$.
\end{pro}
\begin{proof}
The idea we put forward to prove this proposition is using conditional density function. Since the $X_i$ are mutual independent with normal distribution so the vector $X=(X_1,\dots,X_n)$ follows a multivariate normal distribution with mean $(\tilde{\mu_1}, \dots, \tilde{\mu_n})$ and covariance matrix $\sigma^2 I_n$ where $I_n$ is the identity matrix of order n. Moreover by conditioning on each of  $N_i$ and using independency between $Y_{ji}$ we obtain
\begin{equation}\label{densityy}
\mathcal{L}\left(\sum_{j=1}^{N_i} Y_{ji}|N_i=k_i\right) = N(k_i \theta_i, k_i\sigma_i^2),
\end{equation}
for each $1\leq i\leq n$. Thus, independency between $N_i$ and $Y_{ji}$  for all $i$ and $j$ yields
\begin{equation}\label{densityyy}
\mathcal{L}\left( H|N_1=k_1,\dots,N_n=k_n \right)= N\left((k_1\theta_1+\dots+k_n\theta_n ), diag(k_1\sigma_1^2, \dots,k_n\sigma_n^2 )\right).
\end{equation}
Conditioning on the random variable $M$ and using the independency between $W_i$ and $M$ gives the following conditional distribution.
\begin{equation}\label{densityyyy}
\mathcal{L}\left(\sum_{i=1}^M W_i| M=m\right)= N(m\mu, mA).
\end{equation}
Putting \eqref{densityyy} and \eqref{densityyyy} together and using independency between $X, H$ and $\sum_{i=1}^M W_i$ provide the conditional distribution of $R$ given $N_1=k_1,\dots,N_n=k_n, M=m$. i.e., 
\begin{equation}\label{densityyyyy}
\mathcal{L}\left(R| N_1=k_1,\dots,N_n=k_n, M=m\right) = N(u,T).
\end{equation}
\eqref{densityyyyy} gives the conditional density of $R$ given $N_1=k_1,\dots,N_n=k_n, M=m$. To get the density function of $R$ we need to multiply the conditional density by the probability functions  associated to each $N_i$ and $M$ and add them up. This 
completes the proof.
\end{proof}
If we follow the same procedure done for Proposition
\ref{density1} and apply it for the model \eqref{optimal5} we can obtain the density function for the vector $R$. 
\begin{rem}
The density function for the model \eqref{optimal5} is 
\begin{equation}\label{density2}
f_R(r)= \sum_{m=0}^{\infty} \left(\frac{e^{-\lambda} \lambda^m}{m!}\right) \frac{e^{-\frac12(r-u)T^{-1}(r-u)'}}{(2\pi)^{\frac n2}|T|^{\frac 12}},
\end{equation}
where $r=(r_1,\dots,r_n) \in \mathbb{R}^n$, $u = \tilde{\mu}+m\mu$ and $T= Q+mA$.
\end{rem} 

In the sequel of this part we provide the Laplace exponents for both models \eqref{optimal2} and \eqref{optimal5}.
\begin{lem}
Consider the multivariate model \eqref{optimal2}. Then the Laplace exponent for the vector $R =  \left( R_1,\dots, R_n\right)$  at $u=(u_1,\dots, u_n)$  is
\begin{equation}\label{Laplace1}
\log\mathbb{E}(e^{-uR}) = -u \tilde{\mu}^t+\frac{\sigma^2}{2} uu^t + \gamma ( e^{-u\mu^t + uAu^t} -1) + \sum_{k=1}^n \lambda_k ( e^{-\theta_ku_k + \frac{\sigma_k^2u_k^2}{2}}-1),
\end{equation}
\noindent where $u^t, \mu^t$ and $\tilde{\mu}^t$  are the column vectors associated to the row vectors $u, \mu$  and $\tilde{\mu} = (\tilde{\mu_1}, \dots, \tilde{\mu_n})$  respectively.
\end{lem}
\begin{proof}
The independency between $X , H$ and $\sum_{k=1}^M W_k$  and using this point that Laplace transform for normal and compound Poisson distributions exists, yield the result.
\end{proof}

\begin{lem}
Consider the multivariate model \eqref{optimal5}. Then the Laplace exponent for the vector $R =  \left( R_1,\dots, R_n\right)$  at $u=(u_1,\dots, u_n)$  is

\begin{equation}\label{Laplace2}
\log\mathbb{E}(e^{-uR}) = -u\tilde{\mu}^t + uQu^t  +  \lambda ( e^{-u\mu^t + uAu^t}-1).
\end{equation}
\end{lem}

\begin{proof}
The Laplace transform for Gaussian distributions and  compound Poisson distributions exists. So, \eqref{Laplace2} can be driven by using the independency between $X$ and $\sum_{k=1}^M W_k$.
\end{proof}

Now, we apply EVaR$_{\alpha}$ along with the model proposed in $(\ref{optimal2})$ to the optimal portfolio problem $(\ref{optimal1})$. Thus, $(\ref{optimal1})$ is written as follows.

\begin{eqnarray}\label{optimal4}
\min_{\omega,s} & &\Big \{\sum_{k=1}^n (-\tilde{\mu_k}\omega_k+\frac{\sigma^2}{2} s  \omega_k^2) +\frac{ \gamma ( e^{-s \sum_{k=1}^n \mu_k\omega_k + s^2\omega A \omega^t} -1) }{s} \nonumber\\
& & \qquad + \frac{\sum_{k=1}^n \lambda_k ( e^{-\theta_k\omega_k s + \frac{s^2\sigma_k^2\omega_k^2}{2}}-1)- \log \alpha}{s}\Big\}\nonumber\\
\text{subject to}& &\sum_{i=1}^n ( \tilde{\mu_i}+ \lambda_i \theta_i+\mu_i\gamma)\omega_i =\mu^*,\nonumber\\
& &\sum_{i=1}^n \omega_i=1,\nonumber\\
& &\omega_i \geq 0, s\geq 0.
\end{eqnarray}

Applying EVaR$_{\alpha}$ and the model \eqref{optimal5} into the optimal portfolio problem $(\ref{optimal1})$ yield to

\begin{eqnarray}\label{optimal6}
\min_{\omega,s} & &\sum_{k=1}^n -\tilde{\mu_k}\omega_k+s \omega Q \omega^t +\frac{ \lambda ( e^{-s \sum_{k=1}^n \mu_k\omega_k + s^2\omega A \omega^t} -1) - \log \alpha }{s} \nonumber\\
\text{subject to}& &\sum_{i=1}^n (\tilde{\mu_i}+ \mu_i\lambda) \omega_i=\mu^*,\nonumber\\
& &\sum_{i=1}^n \omega_i=1,\nonumber\\
& &\omega_i \geq 0, s\geq 0.
\end{eqnarray}

\subsection{Necessary and Sufficient Conditions for Optimal Problems, KKT Conditions}

In this section we would like to identify the necessary and sufficient conditions for optimality of problems $(\ref{optimal4})$ and $(\ref{optimal6})$. In fact, we want to examine the  Karush-Kuhn-Tucker (KKT) conditions for these problems and study whether the constrained problems in the last two sections have optimal solutions. Being the objective functions for both problems  $(\ref{optimal4})$ and $(\ref{optimal6})$ smooth enough (they are continuously differentiable functions), will help us to verify the KKT conditions much easier.  
\subsubsection{KKT Conditions for Optimal Problem with the multivariate model 1}

The KKT conditions provide necessary conditions for a point to be optimal point for a constrained nonlinear optimal problem.  We refer to Chapter 5 page 241 \cite{Belegundu} for a comprehensive study of KKT conditions for nonlinear optimal problems. Here, we study these conditions for the model $(\ref{optimal4})$ by using the same notation used in page 200 \cite{Belegundu}. We rewrite problem $(\ref{optimal4})$ as follows.

\begin{eqnarray}\label{optimal7}
\min_{\omega,s} & & f(\omega,s)= \Big \{\sum_{k=1}^n (-\tilde{\mu_k}\omega_k+\frac{\sigma^2}{2} s  \omega_k^2) +\frac{ \gamma ( e^{-s \sum_{k=1}^n \mu_k\omega_k + s^2\omega A \omega^t} -1) }{s} \nonumber\\
& & \qquad + \frac{\sum_{k=1}^n \lambda_k ( e^{-\theta_k\omega_k s + \frac{s^2\sigma_k^2\omega_k^2}{2}}-1)- \log \alpha}{s}\Big\}\nonumber\\
\text{subject to}& &h_1(\omega,s)= \sum_{i=1}^n ( \tilde{\mu_i}+ \lambda_i \theta_i+\mu_i\gamma) \omega_i-\mu^*=0,\nonumber\\
& &h_2(\omega,s)=\sum_{i=1}^n \omega_i-1 =0,\nonumber\\
& &g_i(\omega,s)=-\omega_i \leq 0,~~~~~ \forall 1\leq i\leq n, \nonumber\\
&& g_{n+1}(\omega,s)= -s\leq 0.
\end{eqnarray}

Let $\omega_s=(\omega,s)$ be a {\it regular point}\footnote{Let $\omega_s$ be a feasible point. Then, $\omega_s$ is said to be a regular point if the gradient vectors $\nabla g_i(\omega_s)$ for $i\in \{ i: g_i(\omega_s) = 0,~~~~ i=1,\dots, n+1\}$ are linearly independent.} for the problem  $(\ref{optimal4})$. Then, the point $\omega_s$ is a local minimum of $f$ subject to the constraints in $(\ref{optimal7})$ if there exists Lagrange multipliers $\nu_1,\dots, \nu_{n+1}$ and $\eta_1,\eta_2$  for the Lagrangian function $L= f(\omega_s) + \sum_{k=1}^{n+1} \nu_k g_k(\omega_s)+ \sum_{j=1}^2 \eta_j h_j(\omega_s)$ such that the followings are true.

\begin{enumerate}
\item  $\frac{\partial L}{\partial \omega_i} = -\tilde{\mu_i} + s\sigma^2 \omega_i + \gamma (-\mu_i + 2s^2(\omega A)_i) \frac{( e^{-s \sum_{k=1}^n \mu_k\omega_k + s^2\omega A \omega^t})}{s} + \lambda_i (-\theta_i + s^2 \sigma^2 \omega_i)\frac{e^{-\theta_k\omega_k s + \frac{s^2\sigma_k^2\omega_k^2}{2}}}{s} -\nu_i +  ( \tilde{\mu_i}+ \lambda_i \theta_i+\mu_i\gamma)\eta_1 +\eta_2 = 0, ~~~ i=1,\dots,n,$
\item $\frac{\partial L}{\partial s} = \sum_{k=1}^n \frac{\sigma^2}{2} \omega_k^2 
+\frac{(-\gamma \sum_{k=1}^n \mu_k\omega_k + \gamma s^2\omega A \omega^t -\gamma)(e^{-s \sum_{k=1}^n \mu_k\omega_k + s^2\omega A \omega^t})+\gamma}{s^2} + \\
\qquad \frac{\sum_{k=1}^n \PAR{\PAR{\lambda_k( -\theta_k\omega_k s + s^2\sigma_k^2\omega_k^2- 1)}\PAR{e^{-\theta_k\omega_k s + \frac{s^2\sigma_k^2\omega_k^2}{2}}}} 
+\sum_{k=1}^n \lambda_k + \log \alpha}{s^2}=0, $
\item $\nu_k \geq 0, ~~~~~ k= 1,\dots, n+1,$
\item $\nu_k g_k = 0, ~~~~~ k= 1,\dots, n+1,$
\item $g_k \leq 0, ~~~~~ k= 1,\dots, n+1, ~\text{and}~ h_j =0, ~~~ j=1,2,$
\end{enumerate}
\noindent where $(\omega A)_i$ is the ith entry of the row vector $(\omega A)$.\\

\begin{rem}
Since, the functions $h_1$ and $h_2$ in \eqref{optimal7} are linear and the functions $g_i$ for $i = 1,\dots, n$ are convex, then by referring to  Section $5.7$ of \cite{Belegundu} we see that the feasible region $\Omega = \{ \omega_s : g_k(\omega_s) \leq 0, ~~ k= 1,\dots, n+1, ~\text{and}~ h_j(\omega_s) =0, ~~~ j=1,2\}$ is a convex set. On the other hand, the risk measure $\rho = EvaR_{1-\alpha}$ is a convex function subject to the variables $\omega_i$ and $s$ for all $i=1,\dots,n$. We refer to \cite{Ahmadi} for a proof. Thus, the objective function $f$ in problem $(\ref{optimal7})$ is convex too.  We see that any {\bf local minimum} for problem $(\ref{optimal7})$ is a {\bf global minimum} too and the KKT conditions are also {\bf sufficient}. See \cite{Belegundu} page 212.
\end{rem}

\subsubsection{KKT Conditions for Optimal Problem with the multivariate 2}
In this section we will provide the KKT conditions for the optimal problem $(\ref{optimal6})$. We show that these conditions are also sufficient for a solution to be an optimal one. First, we rewrite the problem $(\ref{optimal6})$ in the following way.

\begin{eqnarray}\label{optimal8}
\min_{\omega,s} & & f(\omega,s)=\sum_{k=1}^n -\tilde{\mu_k}\omega_k+s \omega Q \omega^t +\frac{ \lambda ( e^{-s \sum_{k=1}^n \mu_k\omega_k + s^2\omega A \omega^t} -1) - \log \alpha }{s} \nonumber\\
\text{subject to}& &h_1(\omega,s)= \sum_{i=1}^n ( \tilde{\mu_i} + \mu_i \lambda) \omega_i-\mu^*=0,\nonumber\\
& &h_2(\omega,s)=\sum_{i=1}^n \omega_i-1 =0,\nonumber\\
& &g_i(\omega,s)=-\omega_i \leq 0,~~~~~ \forall 1\leq i\leq n, \nonumber\\
&& g_{n+1}(\omega,s)= -s\leq 0.
\end{eqnarray}

By applying the same definition and notation used in the previous section we can provide the KKT conditions as follows.

\begin{enumerate}
\item  $\frac{\partial L}{\partial \omega_i} = -\tilde{\mu_i} + 2s (\omega Q)_i + \lambda (-s \mu_i + 2s^2 (\omega A)_i)\frac{( e^{-s \sum_{k=1}^n \mu_k\omega_k + s^2\omega A \omega^t})}{s} = 0 , ~~~ i=1,\dots,n,$
\item $\frac{\partial L}{\partial s} = \omega A \omega^t + \frac{\lambda (- \sum_{k=1}^n \mu_k\omega_k + 2s^2\omega A \omega^t -1) (e^{-s \sum_{k=1}^n \mu_k\omega_k + s^2\omega A \omega^t}) + \lambda + \log\alpha}{s^2} = 0, $
\item $\nu_k \geq 0, ~~~~~ k= 1,\dots, n+1,$
\item $\nu_k g_k = 0, ~~~~~ k= 1,\dots, n+1,$
\item $g_k \leq 0, ~~~~~ k= 1,\dots, n+1, ~\text{and}~ h_j =0, ~~~ j=1,2,$
\end{enumerate}
\noindent where $(\omega A)_i$ and $(\omega Q)_i$ are the ith entry of the row vectors $(\omega A)$ and $(\omega Q)$ respectively.\\

\begin{rem}
Since the feasible region $\Omega = \{ \omega_s : g_k(\omega_s) \leq 0, ~~ k= 1,\dots, n+1, ~\text{and}~ h_j(\omega_s) =0, ~~~ j=1,2\}$ and the objective function for the optimal problem \eqref{optimal8} are convex, so again by referring to \cite{Belegundu} we can see that the KKT conditions are  also {\bf sufficient} and any {\bf local minimum} for problem $(\ref{optimal8})$ is a {\bf global minimum} as well.
\end{rem}

\section{Efficient Frontier Analysis}
In this section we study the optimization problem \eqref{optimal1} for multivariate model 1 given in $(\ref{optimal2})$. In fact, we analyze the efficient frontier for this problem when the risk measures are EVaR and standard deviation.  Our analysis shows that we have different portfolio decomposition corresponding to EVaR and standard deviation as the underlined model for returns is followed a non-elliptical distribution(model 1). Thanks to the closed form for EVaR we can use optimization packages in mathematical software to solve the optimization problem \eqref{optimal4} without using simulation techniques like Monte Carlo simulation. 
\subsection{Parameters Estimation}
Studying the optimization problems \eqref{optimal4} and \eqref{optimal6} requires knowing the parameters of the multivariate models \eqref{optimal2} and \eqref{optimal5}. To estimate these parameters we use 
a method of estimation for joint parameters so called {\it Extended Least Square(ELS)}\cite{Vonesh}. In fact, assume that we are given a sample of $n$ individuals. Let $y_i = [y_{i1}. \dots, y_{i p_i }]$ denote the $i^{th}$ subject's $1\times p_i$ vector of repeated measurements where  the $y_i$ are assumed to be independently distributed with mean and covariance matrices given by 
\begin{eqnarray}\label{optimal9}
\mathbb{E}(y_i) &=& {\bf \bar{\mu}_i(\beta)}\\\nonumber
Cov(y_i) &=& G_i(\beta, \theta),
\end{eqnarray}
where $\beta$ and $\theta$ are vectors of unknown parameters which should be estimated. Extended Least Square(ELS) estimates are obtained by minimizing the following objective function.
\begin{equation}\label{optimal10}
f(\beta, \theta)=\sum_{i=1}^n \{ (y_i-\bar{\mu}_i(\beta)) G_i^{-1}(\beta,\theta)  (y_i-\bar{\mu}_i(\beta))' + \log|G_i(\beta,\theta)|\},
\end{equation}
where $ \bar{\mu}_i(\beta)$ and $G_i(\beta, \theta)$ are defined in \eqref{optimal9} and $|G_i|$ is the determinant of the positive definite covariance matrix $G_i$. Following \cite{Vonesh} it can be seen that ESL is joint normal theory maximum likelihood estimation. In fact, minimizing \eqref{optimal10} is equivalent to maximizing the log-likelihood function of the $y_i$ when the $y_i$ are independent and normally distributed with mean anc covariance matrices given by \eqref{optimal9}.

\subsection{Data Sets}
We construct the portfolio by choosing 3 stocks. They are APPLE, INTEL and PFIZER(PFE). We use the close data ranged from 20/09/2010 to 26/08/2013. The weekly close data are converted to log return. i.e., if we consider $P_n$ as the close price for the week $n^{th}$ then log return is $R_n= \ln P_n - \ln P_{n-1}$. 

Now consider the model \eqref{optimal2}. We try to apply this model to these three stocks and determine the parameters in \eqref{optimal9} in order to solve the optimization problem \eqref{optimal10}. In this case we have $n=153$, the number of our sample and $y_i$ is a $1\times 3$ row vector associated to the mean of returns. 
Then the vector $\bar{\mu_i}$ is,
\begin{equation}\label{optimal11}
\bar{\mu_i}= \left( \tilde{\mu_1}+\lambda_1\theta_1 + \mu_1 \gamma,  \tilde{\mu_2}+\lambda_2\theta_2 + \mu_2 \gamma,  \tilde{\mu_3}+\lambda_3\theta_3 + \mu_3 \gamma\right),
\end{equation}
for all $1\leq i\leq 153$. Let $A= (a_{ij})_{3\times 3}$ be the covariance matrix for the multivariate normal distribution $W_k$. Then, the covariance matrix $G_i$ in \eqref{optimal9} has the following representation.
\begin{equation}\label{optimal12}
G_i = \left(\begin{array}{ccc}
  \sigma^2+\lambda_1(\theta_1^2+\sigma_1^2)+
  \gamma(a_{11}+ \mu_1^2) & \gamma(a_{12}+\mu_1\mu_2) & \gamma(a_{13}+\mu_1\mu_3) \\
  \gamma(a_{12}+\mu_1\mu_2)& \sigma^2+\lambda_2(\theta_2^2+\sigma_2^2)+
  \gamma(a_{22}+ \mu_2^2) & \gamma(a_{23}+\mu_2\mu_3)\\
  \gamma(a_{13}+\mu_1\mu_3)& \gamma(a_{23}+\mu_2\mu_3) & \sigma^2+\lambda_3(\theta_3^2+\sigma_3^2)+
  \gamma(a_{33}+ \mu_3^2)
\end{array}\right),
\end{equation}
for all $1\leq i\leq 153$. Therefore, by plugging \eqref{optimal11} and \eqref{optimal12} into \eqref{optimal10} we get the objective function for the ELS method. Doing the same procedure for the model \eqref{optimal5} we can find the parameters in \eqref{optimal9}. Let  $Q=(q_{ij})_{3\times 3}$ and $A= (a_{ij})_{3\times 3}$ be the covariance matrices for the multivariate normal distribution $X$ and $W_k$ respectively. Then we have,

\begin{equation}\label{optimal13}
\bar{\mu_i}= \left( \tilde{\mu_1} + \mu_1 \lambda,  \tilde{\mu_2}+ \mu_2 \lambda,  \tilde{\mu_3} + \mu_3 \lambda\right),
\end{equation}
and 
\begin{equation}\label{optimal14}
G_i = \left(\begin{array}{ccc}
  q_{11}+
  \lambda(a_{11}+ \mu_1^2) & q_{12}+\lambda(a_{12}+\mu_1\mu_2) & q_{13}+\lambda(a_{13}+\mu_1\mu_3) \\
 q_{12}+ \lambda(a_{12}+\mu_1\mu_2)& q_{22}+
  \lambda(a_{22}+ \mu_2^2) & q_{23}+\lambda(a_{23}+\mu_2\mu_3)\\
  q_{13}+\lambda(a_{13}+\mu_1\mu_3)& q_{23}+\lambda(a_{23}+\mu_2\mu_3) & q_{33}+
  \lambda(a_{33}+ \mu_3^2)
\end{array}\right),
\end{equation}
for all $1\leq i\leq 153$.

In the following we provide the results for the portfolio decomposition corresponding to the three stocks, EVaR$_{95\%}$ and standard deviation. This results have been driven for the model 1 given in $(\ref{optimal2})$.  In order to estimate our parameters for the model 1 we call {\it fminsearch} in Matlab, where the function to be optimized is the objective function introduced in \eqref{optimal10}. To find the efficient frontiers of EVaR$_{95\%}$ we also call {\it fmincon} in Matlab, where the function to be optimized is the objective function in \eqref{optimal4}.  Figure \ref{fig:minipage2}  shows the two efficient frontiers based on model 1 for EVaR$_{95\%}$ and standard deviation. Table \ref{table:compositionEVaR} and \ref{table:compositiondeviation} show the portfolio compositions and the corresponding EVaR$_{95\%}$ and standard deviation respectively. 

\begin{figure}[ht]
\centering
\begin{minipage}[b]{0.45\linewidth}
\includegraphics[scale=0.40]{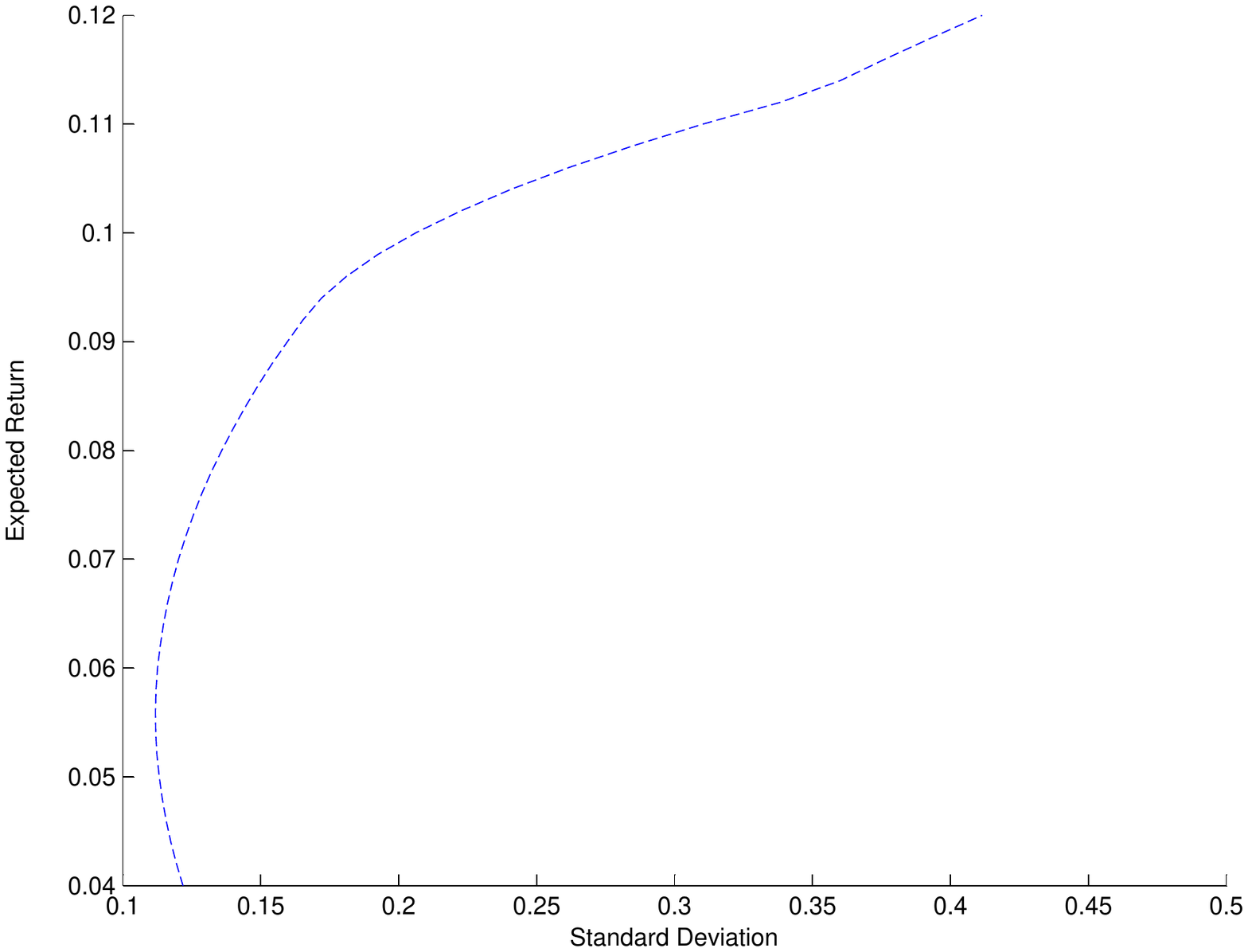}
\label{fig:minipage1}
\end{minipage}
\quad
\begin{minipage}[b]{0.45\linewidth}
\includegraphics[scale=0.45]{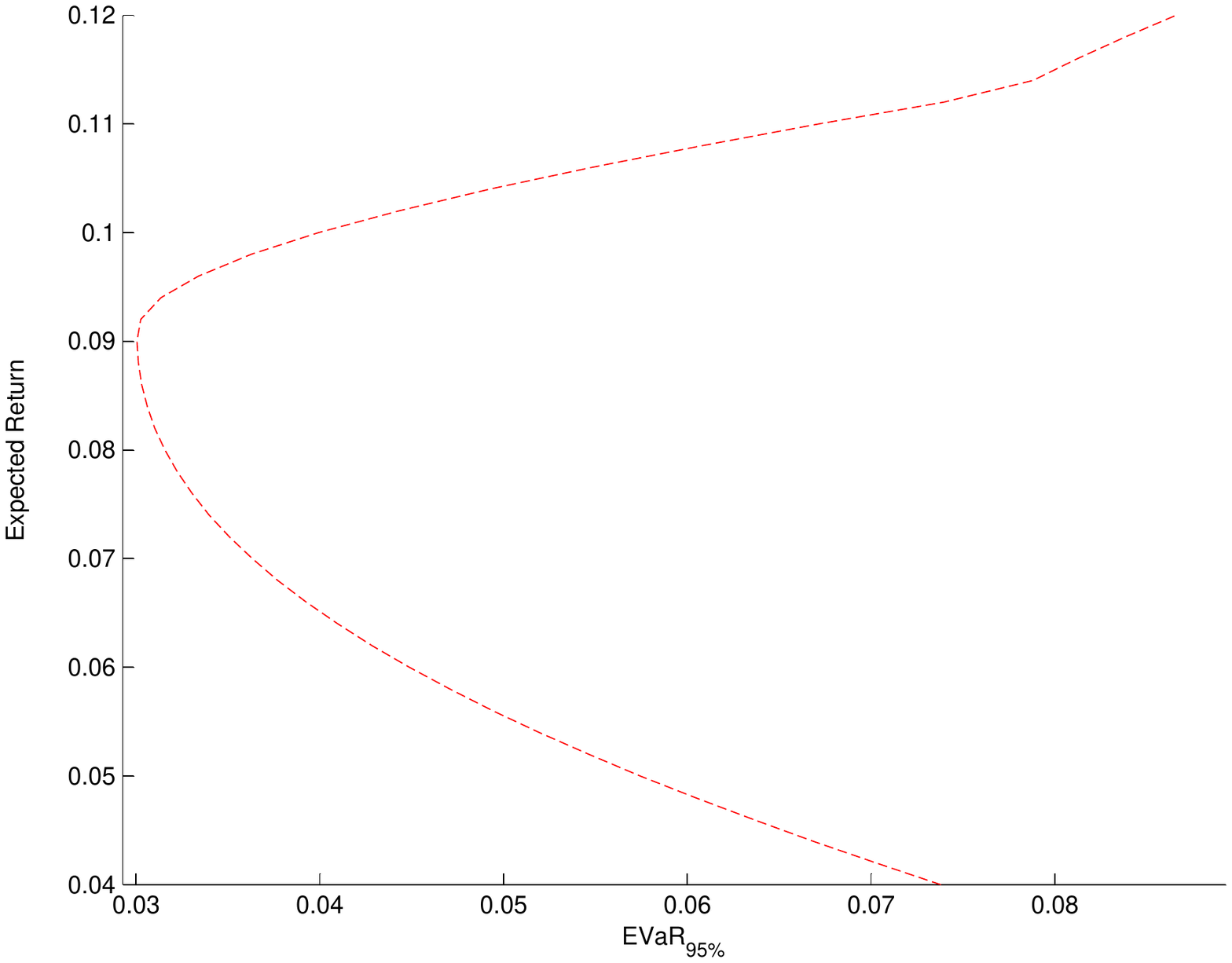}
\caption{Non-elliptical model 1 efficient frontier versus standard deviation and EVaR$_{95\%}$}
\label{fig:minipage2}
\end{minipage}
\end{figure}

%

$$\begin{tabular}{|c|c|c|c|c|c|}
  \hline
  \hbox{Return}& EvaR$_{95 \%}$& \hbox{Apple}& \hbox{Intel}& \hbox{PFE}\\\hline 
   0.0400 & 0.0738 & 0.2743 & 0.4140 &0.3117  \\\hline 
   0.0480 &  0.0604&0.3210  &0.3482 & 0.3308 \\\hline 
   0.0560& 0.0494& 0.3682&0.2827
 &0.3491 \\\hline 
   0.0640& 0.0410& 0.4159& 0.2175& 0.3667\\\hline 
   0.0720&0.0351& 0.4638& 0.1524& 0.3838\\\hline 
  0.0800&0.0316&0.5120 &0.0875& 0.4005 \\\hline 
   0.0880&0.0301 &0.5602 & 0.0226& 0.4172\\\hline 
  0.0960&0.0334& 0.6772&0.0000  & 0.3228\\\hline 
  0.1040& 0.0493& 0.8308 & 0.0000 & 0.1692\\\hline 
   0.1120& 0.0740& 0.9844& 0.0000& 0.0156\\
  \hline
  \end{tabular}$$ \captionof{table}{Portfolio composition and corresponding EvaR$_{95 \%}$}\label{table:compositionEVaR}

$$\begin{tabular}{|c|c|c|c|c|c|}
  \hline
  \hbox{Return}& Deviation& \hbox{Apple}& \hbox{Intel}& \hbox{PFE}\\\hline 
   0.0400 & 0.1219 & 0.2674 & 0.4098 &0.3228 \\\hline 
   0.0480 & 0.1143&0.3194&0.3472& 0.3333 \\\hline 
   0.0560& 0.1119& 0.3715&0.2847
 &0.3438 \\\hline 
   0.0640& 0.1147& 0.4235& 0.2222&0.3543\\\hline 
   0.0720&0.1227&0.4755&0.1596&0.3648\\\hline 
  0.0800&0.1359&0.5276&0.0971& 0.3754 \\\hline 
   0.0880&0.1543 &0.5797&0.0346
& 0.3857\\\hline 
  0.0960&0.1811& 0.6772&0.0000  & 0.3228\\\hline 
  0.1040& 0.2409& 0.8308 & 0.0000 & 0.1692\\\hline 
   0.1120& 0.3386& 0.9844& 0.0000& 0.0156\\
  \hline
\end{tabular}$$
  \captionof{table}{Portfolio composition and corresponding standard deviation}\label{table:compositiondeviation}

\end{document}